\documentclass[final,12pt]{article}
\usepackage{amsfonts,color,morefloats,pslatex,a4wide}
\usepackage{amssymb,amsthm,amsmath,latexsym,pslatex,cite}

\newtheorem{theorem}{Theorem}
\newtheorem{lemma}[theorem]{Lemma}

\newtheorem{open}[theorem]{Open Problem}

\usepackage[para]{threeparttable}

\newcommand{\ord}{{\mathrm{ord}}}

\newcommand{\lcm}{{\mathrm{lcm}}}

\newcommand{\gf}{{\mathrm{GF}}}

\newcommand{\wt}{{\mathtt{wt}}}

\newcommand{\PGRM}{{\mathtt{PGRM}}}

\newcommand{\Z}{\mathbb{{Z}}}

\newcommand{\m}{\mathbb{M}}

\newcommand{\cB}{{\mathcal{B}}}

\newcommand{\C}{{\mathcal{C}}}

\newcommand{\cS}{{\mathcal{S}}}

\newcommand{\bc}{{\mathbf{c}}}

\newcommand{\BCH}{{\mathrm{BCH}}}

\makeatletter

\newcommand{\Rmnum}[1]{\expandafter\@slowromancap\romannumeral #1@}
\makeatother

\begin{document}

\title{The Exact Parameters of A Family of BCH Codes 
\thanks{Z. Sun's research was supported by was supported by Anhui Provincial Natural Science Foundation under Grant Number 2408085MA014 and the Fundamental Research Funds for the Central Universities under Grant Number JZ2024HGTB0203 and the National Natural Science Foundation of China under Grant Number U21A20428.}}

\author{Zhonghua Sun 
\thanks{Zhonghua Sun is with the School of Mathematics, Hefei University of Technology, Hefei, Anhui, 230601, China (e-mail:  sunzhonghuas@163.com).}
}

\maketitle

\begin{abstract}
Despite the theoretical and practical significance of BCH codes, the exact minimum distance and dimension remain unknown for many families. This paper establishes the precise minimum distance and dimension of narrow-sense BCH codes $\C_{(q, m, \lambda, \ell_0, \ell_1)}$ over $\gf(q)$ of length $\frac{q^m-1}{\lambda}$ and designed distance $\frac{(q-\lambda \ell_0)q^{m-1-\ell_1}-1}{\lambda}$, where $\lambda\mid (q-1)$, $0\leq \ell_0< \frac{q-1}{\lambda}$, and $0\leq \ell_1\leq m-1$. These results conclusively resolve the three open problems posed by Li et al. (IEEE Trans. Inf. Theory, vol. 63, no. 11, pp. 7219-7236, Nov. 2017) while establishing complementary advances to Ding's seminal framework (IEEE Trans. Inf. Theory, vol. 61, no. 10, pp. 5322-5330, Oct. 2015).

\vspace*{.3cm}


\end{abstract}

%\newpage 
\section{Introduction and Motivation}

Let $q$ be a prime power and let $\gf(q)$ be the finite field of order $q$. An $[ n, k, d]$ linear code $\C$ over $\gf(q)$ is a $k$-dimensional subspace of $\gf(q)^n$ with minimum distance $d$. The code $\C$ is called {\it cyclic} if $(c_0,c_1,\ldots, c_{n-1})\in\C$ implies $(c_{n-1},c_0,\ldots,c_{n-2})\in \C$. Define
\begin{align*}
\varphi:~\gf(q)^n &\rightarrow R:=\gf(q)[x]/\langle x^n-1\rangle\\
(c_0,c_1,\ldots,c_{n-1})&\mapsto c_0+c_1x+\cdots+c_{n-1}x^{n-1}.
\end{align*}
Let $\C\subseteq \gf(q)^n$, define $\varphi(\C)=\{\varphi(\bc):~\bc\in \C\}$. In the paper, we identify $\C$ with $\varphi(\C)$. The code $\C$ is cyclic if $\C$ is an ideal of the ring $R$. It is well known that every ideal of $R$ is the principal. Let $\C=\langle g(x)\rangle$ be a cyclic code over $\gf(q)$ of length $n$, where $g(x)$ is monic and has the smallest degree among all the generators of $\C$. Then $g(x)$ is unique and is called the generator polynomial, and $h(x)=(x^n-1)/g(x)$ is called the parity check polynomial of $\C$.

Let $n$ be a positive integer with $\gcd(n, q)=1$.  Denote by $\ord_n (q)$ the multiplicative order of $q$ modulo $n$. Set $m=\ord_{n}(q)$. Then there exists an integer $\lambda\geq 1$ such that $n=(q^m-1)/\lambda$, and define $N=n\lambda$. Let $\alpha$ be a primitive element of $\gf(q^m)$. For $0\leq i\leq N-1$, the \textit{ minimal polynomial} $\m_{\alpha^i}(x)$ of $\alpha^i$ over $\gf(q)$ is the monic polynomial of the smallest degree over $\gf(q)$ with $\alpha^i$ as zero. Let $\beta=\alpha^\lambda$. Then $\beta$ is a primitive $n$-th root of unity. Let $2 \leq \delta \leq n$. A \emph{BCH cyclic code} over $\gf(q)$ of length $n$ and \emph{designed distance} $\delta$ with respect to $\beta$, denoted by $\C_{(q, m, \lambda, \delta)}$, is a cyclic code over $\gf(q)$ of length $n$ with generator polynomial 
\begin{eqnarray}\label{eqn-BCHdefiningSet}
\lcm(\m_{\beta^{1}}(x), \m_{\beta^{2}}(x), \cdots, \m_{\beta^{\delta-1}}(x)),  
\end{eqnarray}
where lcm denotes the least common multiple of these polynomials. The code $\C_{(q, m, \lambda, \delta)}$ is \emph{primitive} if $\lambda=1$, and \emph{non-primitive} otherwise.

BCH codes, introduced independently by Hocuenghem \cite{BCH1}, Bose and Ray-Chaudhuri \cite{BCH2}, have significant theoretical and practical importance. A fundamental open problem is determining their dimensions and true minimum distances \cite{Charpin}. Although the dimensions of various BCH codes have been extensively studied \cite{Diem1,CJLM,DDZ,DLMQ,GLQS,HYWSM,BCH4,LLDL,BCH6,LLFLR,LLGS,LYW,PZK,WWLW,ZSK}, the minimum distance remains unresolved, in general. The minimum distances of primitive BCH codes were determined \cite{Min1,BST,Ding,DDZ,DFZ,KL,LI,LWL,BCH7,SB}; however, the results for non-primitive BCH codes are very limited \cite{BCH5,BCH7,NLJM,XWLC,ZSWH,ZSK}. For more information on the BCH codes, we refer the reader to \cite{DL}.   

The main objective of this paper is to determine the dimension and minimum distance of the code $\C_{(q, m, \lambda, \ell_0, \ell_1)}$ with the designed distance $$\delta=\frac{(q-\lambda \ell_0)q^{m-1-\ell_1}-1}{\lambda},$$ for integer pairs $(\ell_0, \ell_1)$, where $\lambda\mid (q-1)$, $0\leq \ell_0< \frac{q-1}{\lambda}$, and $0\leq \ell_1\leq m-1$. These flexible codes generalize many BCH codes. Previous work includes: 
\begin{enumerate}
    \item For $\lambda=1$, Ding established good parameters and solved the minimum distance of $\C_{(q, m, \lambda, \ell_0, \ell_1)}$, but left its dimension incompletely determined \cite{Ding}. 
    \item For $\lambda=q-1$, Li et al. demonstrated very good parameters and posed three open problems \cite{BCH5}:
 \begin{open}\cite[Open Problem 37]{BCH5}\label{Open1}
 Dose the code $\C_{(q, m, q-1, 0, \ell_1 )}$ have minimum distance $(q^{m-\ell_1}-1)/(q-1)$, where $1\leq \ell_1\leq  m-2$?
 \end{open}
 \begin{open}\cite[Open Problem 43]{BCH5}\label{Open2}
 Dose the code $\C_{(q, m, q-1,0, m-2 )}$ have minimum distance $q+1$? 
 \end{open}
 \begin{open}\cite[Open Problem 47]{BCH5}\label{Open3}
    Determine the dimension of $\C_{(q, m, q-1,0, \ell_1 )}$ for $1\leq \ell_1\leq  \lfloor (m-2)/2\rfloor$.
    \end{open}
\end{enumerate}
Our results completely resolve these open problems. We assess code optimality by comparing with the best known linear codes in Grassl's database \cite{Grassl}.

\section{Punctured generalized Reed-Muller codes}\label{sec2}

Throughout this paper, we fix the following notation, unless it is stated otherwise:
\begin{enumerate}
    \item $q$ is a prime power.
    \item $m\geq 2$ is an integer.
	\item $\lambda\geq 1$ and $\lambda \mid (q-1)$.
	\item $N=q^m-1$.
	\item $\alpha$ is a primitive element of $\gf(q^m)$.
	\item $\beta=\alpha^\lambda$ is a primitive $n$-th root of unity.
\end{enumerate}
Let $\Z_{N}=\{0,1,2,\cdots, N-1\}$ denote the ring of integers modulo $N$. For any integer $i$, let $i\bmod N$ denote the unique $j\in \Z_N$ such that $i\equiv j\pmod{N}$. For $i \in \Z_{N}$, the \emph{$q$-cyclotomic coset of $i$ modulo $N$} is defined by 
\[C^{(q,N)}_i=\{iq^j \bmod {N}:~0\leq j\leq \ell_i-1\} , \]
where $\ell_i$ is the smallest positive integer such that $i \equiv i q^{\ell_i} \pmod{N}$, and is the \textit{size} of the $q$-cyclotomic coset $C^{(q, N)}_i$. The smallest integer in $C^{(q, N)}_i$ is its \textit{coset leader}. Let $\Gamma_{(q, N)}$ be the set of all coset leaders. Then $\Z_N$ partitions as
$$\Z_N=\bigsqcup_{i\in \Gamma_{(q, N)}} C_i^{(q, N)},$$
where $\sqcup$ denotes disjoint union.

 For $0\leq i\leq N$, the \textit{$q$-adic expansion} of $i$ is defined by $i=i_{m-1}\ q^{m-1}+\cdots +i_1\ q+i_0$, where $0\leq i_j\leq q-1$. The \textit{$q$-weight} of $i$, denoted by $\wt_q(i)$, is defined by $\wt_q(i)=i_0+i_1+\cdots+i_{m-1}$. It can be easily verified that $\wt_q(j)=\wt_q(i)$ for all $j\in C_i^{(q, N)}$, that is, the value $\wt_q(i)$ is invariant on $C_i^{(q, N)}$.
 
For $0\leq  \ell<(q-1)m$ and $\lambda \mid \ell$, define the polynomial
\begin{equation}\label{CRM}
g_{(q, m, \lambda,\ell)}(x)= \prod_{1\leq i \leq  n-1  \atop \wt_q(\lambda i) < q-1)m-\ell}(x-\alpha^{\lambda i}).	
\end{equation}
Since $\wt_q(\lambda j)=\wt_q(\lambda i)$ for all $j \in C_i^{(q, n)}$,  $g_{(q, m, \lambda,\ell)}(x)\in \gf(q)[x]$. The \textit{$\ell$-th order punctured generalized Reed-Muller code} $\PGRM_q(\ell,m)$ is the cyclic code over $\gf(q)$ of length $n$ with generator polynomial $g_{(q, m, \lambda,\ell)}(x)$. 
\begin{enumerate}
    \item When $\lambda=1$, $\PGRM_q(\ell,m)$ is said to be {\it primitive} \cite{AK98}.
    \item When $\lambda \mid (q-1)$ and $\lambda >1$, $\PGRM_q(\ell,m)$ is said to be {\it non-primitive}  \cite{DGM70}.
\end{enumerate}
Its minimum distance is characterized as follows:

\begin{lemma}\cite[Theorem 5.5.2]{AK98}\cite[Theorem 3.5.1]{DGM70} \label{LEM:5}
Let $\ell=(q-1)\ell_1+\lambda \ell_0$, where $0\leq \ell_1\leq m-1$ and $0\leq \ell_0<\frac{q-1}{\lambda}$. Then
$$d(\PGRM_q(\ell,m))=\frac{(q-\lambda \ell_0)q^{m-1-\ell_1}-1}{\lambda}.$$
\end{lemma}

\section{Parameters of the BCH  code $\C_{(q, m, \lambda, \ell_0, \ell_1)}$}

For integers $0 \leq \ell_1 \leq m-1$ and $0 \leq \ell_0 < \frac{q-1}{\lambda}$, define 
\begin{equation}\label{Delta}
\delta = \frac{(q - \lambda \ell_0) q^{m-1-\ell_1} - 1}{\lambda}.
\end{equation}
To ensure $\delta \geq 2$, we impose the additional constraint that when $\ell_1 = m-1$, $\ell_0$ satisfies $0 \leq \ell_0 \leq \frac{q-1}{\lambda} - 2$.

Let $\C(q, m, \lambda, \ell_0,\ell_1)$ be the BCH code over $\gf(q)$ of length $n$ and designed distance $\delta$ with respect to the primitive $n$-th root of unity $\beta$. By definition, the generator polynomial of $\BCH(q, m, \lambda, \ell_0,\ell_1)$ is 
\begin{equation}\label{gBCHg}
	g(x)=\lcm(\m_{\alpha^{\lambda}}(x), \m_{\alpha^{\lambda \cdot 2}}(x), \cdots, \m_{\alpha^{\lambda \cdot (\delta-1)}}(x)). 
\end{equation}

\subsection{Minimum distance of the BCH code $\C_{(q, m, \lambda, \ell_0, \ell_1)}$}

We now resolve two open problems from \cite{BCH5} by establishing the exact minimum distance of $\C_{(q, m, \lambda, \ell_0, \ell_1)}$. The following theorem connects the code to projective generalized Reed-Muller codes and determines its minimum distance.

\begin{theorem}\label{Distance}
	Let $\delta$ be defined as in \eqref{Delta} and $\ell=(q-1)\ell_1+\lambda \ell_0$. Then $\PGRM_q(\ell,m)\subseteq \C(q, m, \lambda, \ell_0,\ell_1)$ and $d(\C(q, m, \lambda, \ell_0,\ell_1))=d(\PGRM_q(\ell,m))=\delta$.
\end{theorem}

\begin{proof}
The case $\lambda = 1$ follows directly from \cite[Theorem 10]{Ding}. For $\lambda > 1$, we proceed by establishing the inclusion $\PGRM_q(\ell,m) \subseteq \C(q, m, \lambda, \ell_0,\ell_1)$. This reduces to showing $g(x) \mid g_{(q, m, \lambda ,\ell)}(x)$, where $g(x)$ and $g_{(q, m, \lambda ,\ell)}(x)$ are defined in \eqref{gBCHg} and \eqref{CRM}, respectively. Equivalently, $\alpha^{\lambda i}$ must be a root of $g_{(q,m,\lambda,\ell)}(x)$ for all $1 \leq i \leq \delta - 1$, i.e., 
	\begin{equation}\label{weight_cond}
	  \wt_q(\lambda i) < (q-1)m - (q-1)\ell_1 - \lambda \ell_0.
	\end{equation}
	First, observe the decomposition:
	\begin{align*}
	\lambda (\delta-1) &= (q-\lambda \ell_0)q^{m-1-\ell_1}-1-\lambda \\
	&= (q-1-\lambda \ell_0) q^{m-1-\ell_1} + (q-1)\sum_{i=1}^{m-2-\ell_1} q^i + (q-1-\lambda).
	\end{align*}
	Consequently, the $q$-weight satisfies
	\begin{align*}
	\wt_q(\lambda(\delta-1)) &= (q-1)m - (q-1)\ell_1 - \lambda \ell_0 - \lambda \\
	&< (q-1)m - (q-1)\ell_1 - \lambda \ell_0.
	\end{align*}
	
	For $1 \leq i \leq \delta-1$, consider the $q$-adic expansion $\lambda i = i_{m-1-\ell_1}q^{m-1-\ell_1} + \cdots + i_1 q + i_0$, with $0 \leq i_j \leq q-1$ for $0 \leq j \leq m-2-\ell_1$ and $0 \leq i_{m-1-\ell_1} \leq q-1 - \lambda \ell_0$ (since $\lambda i \leq \lambda(\delta - 1)$). We analyze two cases:
	\begin{itemize}
		\item \textit{Case 1}: $0 \leq i_{m-1-\ell_1} < q-1 - \lambda \ell_0$. Then
		\begin{align*}
			\wt_q(\lambda i) &\leq (q-1)(m-1-\ell_1) + i_{m-1-\ell_1} \\
			&< (q-1)(m-1-\ell_1) + q-1 - \lambda \ell_0 \\
			&= (q-1)m - (q-1)\ell_1 - \lambda \ell_0.
		\end{align*}
		
		\item \textit{Case 2}: $i_{m-1-\ell_1} = q-1 - \lambda \ell_0$. Let $A = |\{0 \leq j \leq m-2-\ell_1 : i_j = q-1 \}|$. Since $\lambda i \leq \lambda(\delta-1)$, we have $A < m-1-\ell_1$. Thus
		\begin{align*}
			\wt_q(\lambda i) &\leq (q-1 - \lambda \ell_0) + (q-1)A + (q-2)(m-1-\ell_1 - A) \\
			&= (q-1)m - (q-1)\ell_1 - \lambda \ell_0 - (m-1-\ell_1 - A) \\
			&< (q-1)m - (q-1)\ell_1 - \lambda \ell_0.
		\end{align*}
	\end{itemize}
	This verifies \eqref{weight_cond}, proving the inclusion $\PGRM_q(\ell,m) \subseteq \C(q, m, \lambda, \ell_0,\ell_1)$.
	
	Combining the BCH bound with this inclusion yields
	\[
	\delta \leq d(\C(q, m, \lambda, \ell_0,\ell_1)) \leq d(\PGRM_q(\ell,m)).
	\]
	By Lemma \ref{LEM:5}, $d(\PGRM_q(\ell,m)) = \delta$, concluding
	\[
	d(\C(q, m, \lambda, \ell_0,\ell_1)) = d(\PGRM_q(\ell,m)) = \delta.
	\]
	This completes the proof.
\end{proof}

When $\lambda = q-1$, Theorem \ref{Distance} yields
$$d(\C_{(q, m, q-1, (q^{m-\ell_1}-1)/(q-1) )})=(q^{m-\ell_1}-1)/(q-1)$$
for $0 \leq \ell_1 \leq m-2$. This resolves Open Problems \ref{Open1} and \ref{Open2} posed by Li et al. \cite{BCH5}.

\subsection{Dimension of the BCH  code $\C_{(q, m, \lambda, \ell_0, \ell_1)}$ }

We next determine the dimension of $\C_{(q, m, \lambda, \ell_0, \ell_1)}$. Set $N = n\lambda = q^m - 1$. For $0 \leq i \leq N$, consider the unique $q$-adic expansion $i = i_{m-1}q^{m-1} + \cdots + i_1 q + i_0$ with $0 \leq i_j \leq q-1$, and define the vector $\overline{i} = (i_{m-1}, i_{m-2}, \ldots, i_0)$. Define $\overline{A} < \overline{B}$ if and only if $A < B$ as integers.  Cyclic shifts satisfy
$$
\overline{iq^j \bmod N} = (i_{m-1-j}, \dots, i_0, i_{m-1}, \dots, i_{m-j}) \quad \text{for} \quad 1 \leq j \leq m-1.
$$
From \eqref{Delta}, the vector representation of $\lambda\delta$ is
\begin{equation}\label{delta2}
\overline{\lambda\delta} = (\underbrace{0,\dots,0}_{\ell_1}, q-1-\lambda\ell_0,\underbrace{q-1,\dots,q-1}_{m-1-\ell_1}).
\end{equation}
The dimension of the BCH code $\C_{(q, m, \lambda, \ell_0, \ell_1)}$ is given by
\begin{equation}\label{eq:dim}
\dim(\C_{(q,m,\lambda,\ell_0,\ell_1)}) = \left| \left\{ 
1 \leq i \leq N : 
\begin{array}{c} 
\overline{iq^j \bmod N} > \overline{\lambda\delta} \\ 
\forall ~0 \leq j \leq m-1, \\ 
\lambda \mid i 
\end{array}
\right\} \right|+\left| C_{\lambda \delta}^{(q, N)}   \right|.
\end{equation}
When $\ell_0=\ell_1=0$, we have $\lambda \delta=N$, implying $\C_{(q, m, \lambda, \ell_0, \ell_1)}$ is an $[n, 1, n]$ repetition code.  

\begin{lemma}\label{lem:orbit}
Let $1\leq \ell_1\leq m-1$ or $\ell_1=0$ and $0< \ell_0<\frac{q-1}{\lambda}$. Then $\left| C_{\lambda\delta}^{(q, N)} \right| = m$.	
\end{lemma}

\begin{proof} 
We verify $\overline{\lambda\delta q^j \bmod N} > \overline{\lambda\delta}$ for $1 \leq j \leq m-1$ via case analysis:

\textit{Case 1:} $\ell_1 = 0$ and $0 < \ell_0 < (q-1)/\lambda$. By \eqref{delta2},
$$\overline{\lambda\delta} = ( q-1-\lambda\ell_0,\underbrace{q-1,\dots,q-1}_{m-1}).$$
For $1\leq j\leq  m-1$,
\begin{align*}
\overline{\lambda\delta q^j \bmod N} &= \big( \underbrace{q-1,\dots,q-1}_{m-j},\ 
q-1-\lambda\ell_0,\ 
\underbrace{q-1,\dots,q-1}_{j-1} \big)\\
&>\overline{\lambda\delta}.
\end{align*}

\textit{Case 2:} $1 \leq \ell_1 \leq m-1$. By \eqref{delta2},
$$
\overline{\lambda\delta} = (\underbrace{0,\dots,0}_{\ell_1}, q-1-\lambda\ell_0,\underbrace{q-1,\dots,q-1}_{m-1-\ell_1}).
$$
For $1\leq j\leq  \ell_1$,
\begin{align*}
\overline{\lambda\delta q^j \bmod N} &= \big( \underbrace{0,\dots,0}_{\ell_1-j},\ 
q-1-\lambda\ell_0,\ 
\underbrace{q-1,\dots,q-1}_{m-1-\ell_1},\ 
\underbrace{0,\dots,0}_{j} \big)\\
&>\overline{\lambda\delta}.
\end{align*}
For $j = \ell_1 + 1 + j'$ with $0 \leq j' \leq m-2-\ell_1$,
\begin{align*}
\overline{\lambda\delta q^j \bmod N} &= \big( \underbrace{q-1,\dots,q-1}_{m-1-\ell_1-j'},\ 
\underbrace{0,\dots,0}_{\ell_1},\ 
q-1-\lambda\ell_0,\ 
\underbrace{q-1,\dots,q-1}_{j'} \big)\\
&>\overline{\lambda\delta}.
\end{align*}
The result follows in both cases. This completes the proof.
\end{proof}

Under Lemma~\ref{lem:orbit}'s conditions, \eqref{eq:dim} simplifies to

\begin{equation}\label{eq:dim2}
	\dim(\C_{(q,m,\lambda,\ell_0,\ell_1)}) = \left| \left\{ 
1 \leq i \leq N : 
\begin{array}{c} 
\overline{iq^j \bmod N} > \overline{\lambda\delta} \\ 
\forall ~0 \leq j \leq m-1, \\ 
\lambda \mid i 
\end{array}
\right\} \right|+m.
\end{equation}

\begin{theorem}\label{THM7}
Let $\ell_1=0$ and $0< \ell_0<\frac{q-1}{\lambda}$. Then 
$ \dim(\C_{(q,m,\lambda,\ell_0,\ell_1)})= \lambda^{m-1}(\ell_0)^{m}+m$.	
\end{theorem}

\begin{proof}
 By \eqref{delta2},
$$\overline{\lambda\delta} = ( q-1-\lambda\ell_0,\underbrace{q-1,\dots,q-1}_{m-1}).$$
Let $\overline{i}=(i_{m-1}, i_{m-2}, \ldots, i_0)$ with $0\leq i_j\leq q-1$. \begin{itemize}
	\item If there exists $0\leq h\leq m-1$ such that $i_h<q-1-\lambda \ell_0$, then $$\overline{i q^{m-1-h} \bmod N}=(i_h, *, \ldots, *)<\overline{\lambda\delta}.$$  
	\item If $q-\lambda \ell_0\leq i_j \leq q-1$ for all $j$, then $\overline{i q^j \bmod N}> \overline{\lambda\delta}$ for all $j$.
	\item If $\exists~ h$ with $i_h = q-1 - \lambda \ell_0$, then $\overline{i q^j \bmod N}\geq \overline{\lambda\delta}$ for all $0\leq j \leq m-1$ if and only if $i_j=q-1$ for $j\neq h$, i.e., $i\in C_{\lambda \delta}^{(q, N)}$.
\end{itemize}
Therefore, 
\begin{align*}
B&=\left\{ 
1 \leq i \leq N : 
\begin{array}{c} 
\overline{iq^j \bmod N} > \overline{\lambda\delta} \\ 
\forall ~0 \leq j \leq m-1, \\ 
\lambda \mid i 
\end{array}
\right\}\\
&=\left\{ 
1 \leq i \leq N : 
\begin{array}{c} 
\overline{i}=(i_{m-1}, i_{m-2}, \ldots, i_0) \\ 
q-\lambda \ell_0 \leq  i_j \leq q-1\text{ for all } j, \\ 
\lambda \mid i 
\end{array}
\right\}.	
\end{align*}
To compute $|B|$, fix $q-\lambda \ell_0 \leq i_{1},\cdots, i_{m-1}\leq q-1$, and let $i_1+i_2+\cdots+i_{m-1}\equiv s~({\rm mod}~\lambda)$, $0\leq s\leq \lambda-1$.  Then $i_0+i_1+\cdots+i_{m-1} \equiv 0~({\rm mod}~\lambda)$ if and only if $i_0\equiv \lambda -s~({\rm mod}~\lambda)$. For each $0\leq s\leq \lambda-1$, 
$$|\{q-\lambda \ell_0\leq i_0\leq q-1:~i_0\equiv \lambda -s~({\rm mod}~\lambda)    \}|=\ell_0.$$
Since there are $(\lambda \ell_0)^{m-1}$ choices for $(i_1, \dots, i_{m-1})$,
 $$B=(\lambda \ell_0)^{m-1} \times \ell_0=\lambda^{m-1}(\ell_0)^{m}.$$ The result follows from \eqref{eq:dim2} and Lemma~\ref{lem:orbit}. This completes the proof.
\end{proof}

The Griesmer bound for an $[n, k, d]$ linear code over $\gf(q)$, as established in \cite{Griesmer}, states that
$$n\geq \sum_{i=0}^{k-1} \left \lceil  \frac{d}{q^i}   \right\rceil.$$
Consider the case where $q>2$, $m\geq 2$, $\lambda =1$, $\ell_0=1$ and $\ell_0=0$. By combining Theorem \ref{Distance} and Theorem \ref{THM7}, the BCH code $\C_{(q, m, \lambda, \ell_0, \ell_1)}$ has parameters $[q^m-1, 1+m, (q-1)q^{m-1}-1]$. It is easily verified that 
$$\sum_{i=0}^{m} \left \lceil  \frac{(q-1)q^{m-1}-1}{q^i}   \right\rceil=(q-1)q^{m-1}-1+\sum_{i=1}^{m-1} (q-1)q^{m-1-i}+1=q^m-1.$$
Therefore, the BCH code $\mathcal{C}_{(q, m, 1, 1, 0)}$ attains the Griesmer bound for all $q > 2$ and $m \geq 2$. Using the computational algebra system MAGMA, we verified the assertions of Theorem \ref{THM7}. The parameters of the BCH code $\mathcal{C}_{(q, m, \lambda, \ell_0, 0)}$ are tabulated in Table \ref{table1}.

\begin{table*}
{\caption{\rm Parameters of $\C_{(q, m, \lambda, \ell_0, 0)}$ for $1\leq \ell_0<(q-1)/\lambda$
}\label{table1}
\begin{center}
\begin{tabular}{ccccccc}\hline
$q$  & $m$ & $\lambda$  &  $\ell_0$  & Parameters & Optimality \\ \hline
$3$ & $2$ &  $1$ & $1$  & $[8,3,5]$ & Optimal\\ \hline
$3$ & $3$ &  $1$ & $1$  & $[26,4,17]$ & Optimal\\ \hline
$3$ & $4$ &  $1$ & $1$  & $[80,5,53]$ & Optimal\\ \hline
$4$ & $2$ &  $1$ & $1$  & $[15,3,11]$ & Optimal\\ \hline
$4$ & $2$ &  $1$ & $2$  & $[15,6,7]$ & $d_{\rm optimal}=8$\\ \hline
$4$ & $3$ &  $1$ & $1$  & $[63,4,47]$ & Optimal\\ \hline
$4$ & $3$ &  $1$ & $2$  & $[63,11,31]$ & $d_{\rm best}=35$\\ \hline
$5$ & $2$ &  $1$ & $1$  & $[24,3,19]$ & Optimal\\ \hline
$5$ & $2$ &  $1$ & $2$  & $[24,6,14]$ & $d_{\rm optimal}=15$\\ \hline
$5$ & $2$ &  $1$ & $3$  & $[24,11,9]$ & $d_{\rm best}=10$\\ \hline
$5$ & $2$ &  $2$ & $1$  & $[12,4,7]$ & $d_{\rm optimal}=8$\\ \hline
$5$ & $3$ &  $1$ & $1$  & $[124,4,99]$ & Optimal\\ \hline
$5$ & $3$ &  $1$ & $2$  & $[124,11,74]$ & $d_{\rm best}=83$\\ \hline
$5$ & $3$ &  $1$ & $3$  & $[124,30,49]$ & $d_{\rm best}=53$\\ \hline
$5$ & $3$ &  $2$ & $1$  & $[62,7,37]$ & $d_{\rm best}=42$\\ \hline
\end{tabular}
\end{center}}
\end{table*}

A \textit{run} in $\overline{i}=(i_{m-1}, i_{m-2}, \dots, i_0)$ is a maximal sequence of consecutive zeros. \textit{Straight runs} do not wrap around, while \textit{circular runs} do. Let $\mathtt{Run}(\overline{i})$ denote the maximal run length (considering circular or straight runs).  For example, $\mathbf{s}=(0,1,0,0,1,0,0)$ has $\mathtt{Run}(\mathbf{s})=3$.

\begin{lemma}\label{lem::8}
Let $1\leq \ell_1\leq m-1$ and $1\leq i\leq N$. If $\overline{i}$ has a run of length $\ell > \ell_1$, then there exists $0\leq j\leq m-1$ such that $\overline{i q^j \bmod N}< \overline{\lambda \delta}$.	
\end{lemma}

\begin{proof}
There exists $0\leq j\leq m-1$ such that
 $$\overline{i q^j \bmod N}=(\underbrace{0,\ldots,0}_{\ell}, *,*,\ldots, *).$$ Since $\ell>\ell_1$, we obtain $\overline{i q^j \bmod N}<\overline{\lambda \delta}$. This completes the proof.
\end{proof}

\begin{lemma}\label{lem::9}
Let $1\leq \ell_1\leq m-1$ and $1\leq i\leq N$. If $\mathtt{Run}(\overline{i})<\ell_1$, then $\overline{i q^j \bmod N}> \overline{\lambda \delta}$ for all $0\leq j\leq m-1$.	
\end{lemma}

\begin{proof}
Suppose $\mathtt{Run}(\overline{i})=\ell$, then the minimal cyclic shift of $\overline{i}$ has the form $$\mathbf{s}=(\underbrace{0,\ldots,0}_{\ell}, i',*,\ldots, *),$$ where  $1\leq i'\leq q-1$. Since $\ell <\ell_1$, we have $\mathbf{s}>\overline{\lambda \delta}$. The desired result follows. This completes the proof.
\end{proof}

Define
$$\cB=\left\{ 1\leq i\leq N:~ \mathtt{Run}(\overline{i})= \ell_1, \begin{array}{c} 
\overline{iq^j \bmod N} > \overline{\lambda\delta} \\ 
\forall ~0 \leq j \leq m-1, \\ 
\lambda \mid i 
\end{array}  \right\},$$
and for $0\leq \ell\leq m-1$, %\cS_\ell &=\left\{ 1\leq i\leq N:~ \mathtt{Run}(\overline{i})\leq \ell  \right\},\\
$$
\cS_{\ell, \lambda}=\left\{ 1\leq i\leq N:~ \mathtt{Run}(\overline{i})\leq \ell,~\lambda\mid i  \right\}.
$$
 Lemma \ref{lem::8} and Lemma \ref{lem::9} imply for $1\leq \ell_1\leq m-1$,
$$\left\{ 
1 \leq i \leq N : 
\begin{array}{c} 
\overline{iq^j \bmod N} > \overline{\lambda\delta} \\ 
\forall ~0 \leq j \leq m-1, \\ 
\lambda \mid i 
\end{array}
\right\}=\cB~\cup \cS_{\ell_1-1}. $$
Thus by \eqref{eq:dim2}, 
\begin{equation}\label{eq:dim3}
	\dim(\C_{(q, m, \lambda, \ell_0, \ell_1)})=|\cB|+|\cS_{\ell_1-1, \lambda}|+m.
\end{equation}

%To determine $|\cS_{\ell_1-1, \lambda}|$, we need the following results.

 \begin{lemma}\cite{mann}\label{LEM::10}
 Let $\C$ be the BCH code over $\gf(q)$ of length $q^m-1$ and designed distance $q^{m-\ell}$ for $1\leq \ell \leq m-1$. Then 
 $$\dim(\C)=q^m-1-\sum_{i=1}^{\lfloor  \frac{m}{\ell+1}    \rfloor} (-1)^{i-1} \frac{m(q-1)^i}{i} \binom{m-i\ell-1}{i-1} q^{m-i(\ell+1)}.$$	
 \end{lemma}

\begin{lemma}\label{LEM:S_ell}
For $0 \leq \ell \leq m-2$,
$$|\cS_{\ell,1}|= q^m-1-\sum_{i=1}^{\lfloor  \frac{m}{\ell+2}    \rfloor} (-1)^{i-1} \frac{m(q-1)^i}{i} \binom{m-i(\ell+1)-1}{i-1} q^{m-i(\ell+2)}.$$
 \end{lemma}
 
\begin{proof}
The BCH code $\C$ over $\gf(q)$ of length $q^m-1$ and designed distance  $q^{m-\ell-1}$ satisfies $\dim(\C) = |\cS_{\ell,1}|$ because:
\begin{itemize}
	\item If $\overline{i}$ has a run of length $\ell'\geq \ell+1$, then there exists $0\leq j\leq m-1$ such that $$\overline{i q^j \bmod N}=(\underbrace{0,\ldots,0}_{\ell'}, *,*,\ldots, *)\leq  \overline{q^{m-\ell-1}-1}.$$
	\item If $\mathtt{Run}(\overline{i})=\ell'\leq \ell$, all cyclic shifts exceed $\overline{q^{m-\ell-1}-1}$.
	\end{itemize}
 The desired result follows from Lemma \ref{LEM::10}. This completes the proof. 
\end{proof}

\begin{lemma}\label{LEM:11}
For $t \geq 1$ and $\lambda \mid (q-1)$,
$$\left|\left\{(i_1,i_2,\ldots, i_t):~\sum_{i=1}^t i_i \equiv 0~({\rm mod}~\lambda), 1\leq i_j\leq q-1 \right \}\right|=\frac{(q-1)^t}{\lambda}.$$	
\end{lemma}

\begin{proof}
The case $\lambda=1$ is trivial. For $\lambda>1$, fix $1\leq i_1,\cdots, i_{t-1}\leq q-1$, and let $$i_1+i_2+\cdots+i_{t-1}\equiv s~({\rm mod}~\lambda)$$ with $0\leq s\leq \lambda-1$. Then $i_1+i_2+\cdots+i_t \equiv 0~({\rm mod}~\lambda)$ if and only if $i_t\equiv \lambda -s~({\rm mod}~\lambda)$. There are exactly $\frac{q-1}{\lambda}$ choices for $i_t$ in  each residue class, giving $(q-1)^{t-1} (\frac{q-1}{\lambda})$ solutions. This completes the proof.
\end{proof}

 \begin{theorem}\label{THM:13}
 Let $\lambda \mid (q-1)$ and $0\leq \ell \leq m-2$. Then
 $$|\cS_{\ell,\lambda}|= \frac{q^m-1}{\lambda}-\left(\frac{q-1}{\lambda}\right)\sum_{i=1}^{\lfloor  \frac{m}{\ell+2}    \rfloor} (-1)^{i-1} \frac{m(q-1)^{i-1}}{i}  \binom{m-i(\ell+1)-1}{i-1}  q^{m-i(\ell+2)}.$$	
 \end{theorem}
 
 \begin{proof}
 Since $\lambda |(q-1)$, we have $q\equiv 1~({\rm mod}~\lambda)$, implying $i\equiv i_{m-1}+i_{m-2}+\cdots+i_0~({\rm mod}~\lambda)$, for $i=i_{m-1}q^{m-1}+\cdots+i_1q+i_0$. Thus $\lambda \mid i$ if and only if $i_{m-1}+i_{m-2}+\cdots+i_0\equiv 0~({\rm mod}~\lambda)$. The constraint $\mathtt{Run}(\overline{i}) \leq \ell$ fixes zero/non-zero patterns. For each pattern, digits in non-zero runs are uniform over $\{1,\cdots,q-1\}$, and by Lemma~\ref{LEM:11}, exactly $\frac{1}{\lambda}$ of these sequences satisfy the sum condition. Hence $|\cS_{\ell,\lambda}| = |\cS_{\ell,1}| / \lambda$. Apply Lemma~\ref{LEM:S_ell} to conclude. This completes the proof. 
\end{proof}

\begin{theorem}\label{THM::14}
Let $\lambda \mid(q-1)$, $1\leq \ell_1\leq m-1$, and $\ell_0=0$. Then
$$\dim(\C_{(q, m, \lambda, \ell_0, \ell_1)})=\frac{q^m-1}{\lambda}-\left(\frac{q-1}{\lambda}\right)\sum_{i=1}^{\lfloor  \frac{m}{\ell_1+1}    \rfloor} (-1)^{i-1} \frac{m(q-1)^{i-1}}{i}  \binom{m-i \ell_1-1}{i-1}  q^{m-i(\ell_1+1)}+m.$$
\end{theorem}

\begin{proof}
 When $\ell_0=0$, \eqref{delta2} gives
$$\overline{\lambda\delta} = ( \underbrace{0,\ldots, 0}_{\ell_1},\underbrace{q-1,\dots,q-1}_{m-\ell_1}).$$
It follows that $\cB=\emptyset$. By \eqref{eq:dim3}, we get $\dim(\C)=|\cS_{\ell_1-1, \lambda}|+m$.
Substitute $|\mathcal{S}_{\ell_1-1,\lambda}|$ from Theorem~\ref{THM:13} with $\ell = \ell_1 - 1$. This completes the proof.
\end{proof}

Using the computational algebra system MAGMA, we verified the assertions of Theorem \ref{THM::14}. The parameters of the BCH code $\mathcal{C}_{(q, m, \lambda, 0, \ell_1)}$ are tabulated in Table \ref{table2}.

\begin{table*}
{\caption{\rm Parameters of $\C_{(q, m, \lambda, 0, \ell_1)}$ for $1\leq \ell_1\leq m-1$
}\label{table2}
\begin{center}
\begin{tabular}{ccccccc}\hline
$q$  & $m$ & $\lambda$  &   $\ell_1$ & Parameters & Optimality \\ \hline
$2$ & $3$ &  $1$ &  $1$ & $[7,4,3]$ & Optimal\\ \hline
$2$ & $4$ &  $1$ &  $1$ & $[15,5,7]$ & Optimal\\ \hline
$2$ & $4$ &  $1$ &  $2$ & $[15,11,3]$ & Optimal\\ \hline
$2$ & $5$ &  $1$ &  $1$ & $[31,6,15]$ & Optimal\\ \hline
$2$ & $5$ &  $1$ &  $2$ & $[31,16,7]$ & $d_{\rm optimal}=8$\\ \hline
$2$ & $5$ &  $1$ &  $3$ & $[31,26,3]$ & Optimal\\ \hline
$2$ & $6$ &  $1$ &  $1$ & $[63,7,31]$ & Optimal\\ \hline
$2$ & $6$ &  $1$ &  $2$ & $[63,24,15]$ & $d_{\rm best}=16$\\ \hline
$2$ & $6$ &  $1$ &  $3$ & $[63,45,7]$ & $d_{\rm best}=8$\\ \hline
$2$ & $6$ &  $1$ &  $4$ & $[63,57,3]$ & Optimal\\ \hline
$3$ & $3$ &  $1$ &  $1$ & $[26,11,8]$ & $d_{\rm best}=9$\\ \hline
$3$ & $3$ &  $2$ &  $1$ & $[13,7,4]$ & $d_{\rm optimal}=5$\\ \hline
$3$ & $4$ &  $1$ &  $1$ & $[80,20,26]$ & $d_{\rm best}=33$\\ \hline
$3$ & $4$ &  $1$ &  $2$ & $[80,60,8]$ & Best Known\\ \hline
$3$ & $4$ &  $2$ &  $1$ & $[40,12,13]$ & $d_{\rm best}=18$\\ \hline
$3$ & $4$ &  $2$ &  $2$ & $[40,32,4]$ & $d_{\rm optimal}=5$\\ \hline
$3$ & $5$ &  $1$ &  $1$ & $[242,37,80]$ & $d_{\rm best}=98$\\ \hline
$3$ & $5$ &  $1$ &  $2$ & $[242,157,26]$ & Best Known\\ \hline
$3$ & $5$ &  $1$ &  $3$ & $[242,217,8]$ & Best Known\\ \hline
$3$ & $5$ &  $2$ &  $1$ & $[121,21,40]$ & $d_{\rm best}=55$\\ \hline
$3$ & $5$ &  $2$ &  $2$ & $[121,81,13]$ & $d_{\rm best}=15$\\ \hline
$3$ & $5$ &  $2$ &  $3$ & $[121,111,4]$ & $d_{\rm optimal}=5$\\ \hline
$4$ & $2$ &  $1$ &  $1$ & $[15,11,3]$ & $d_{\rm optimal}=4$\\ \hline
$4$ & $3$ &  $1$ &  $1$ & $[63,30,15]$ & $d_{\rm best}=18$\\ \hline
$4$ & $3$ &  $1$ &  $2$ & $[63,57,3]$ & $d_{\rm optimal}=4$\\ \hline
$4$ & $3$ &  $3$ &  $1$ & $[21,12,5]$ & $d_{\rm optimal}=7$\\ \hline
$4$ & $4$ &  $3$ &  $1$ & $[85,31,21]$ & $d_{\rm best}=29$\\ \hline
$4$ & $4$ &  $3$ &  $2$ & $[85,73,5]$ & $d_{\rm best}=6$\\ \hline
$5$ & $2$ &  $1$ &  $1$ & $[24,18,4]$ & $d_{\rm optimal}=5$\\ \hline
$5$ & $2$ &  $2$ &  $1$ & $[12,10,2]$ & Optimal\\ \hline
$5$ & $3$ &  $2$ &  $1$ & $[62,35,12]$ & $d_{\rm best}=15$\\ \hline
$5$ & $3$ &  $2$ &  $2$ & $[62,59,2]$ & Optimal\\ \hline
$5$ & $3$ &  $4$ &  $1$ & $[31,19,6]$ & $d_{\rm best}=8$\\ \hline
\end{tabular}
\end{center}}
\end{table*}

% For $1\leq \ell \leq m-1$, it is easily verified that 
% \begin{align*}
 %\frac{m}{i} \binom{m-i\ell-1}{i-1}&=\frac{m-i\ell+i\ell}{i}\binom{m-i\ell-1}{i-1}\\
% &=\frac{m-i\ell}{i}\binom{m-i\ell-1}{i-1}+\ell 	\binom{m-i\ell-1}{i-1}\\
%& =\binom{m-i\ell}{i}+\ell 	\binom{m-i\ell-1}{i-1},
% \end{align*}
% for $1\leq i\leq \lfloor  \frac{m}{\ell+1}    \rfloor$.

%Let $1\leq \ell_1\leq m-1$ and $0<\ell_0<\frac{q-1}{\lambda}$, where $\lambda \mid(q-1)$. 

%and 
%\begin{itemize}
%	\item if $x_1+x_{t+1}=\ell_1$, $q-\lambda \ell_0 \leq a_1\leq q-1$;
%	\item if $x_j=\ell_1$ for $2\leq j\leq t$, $q-\lambda \ell_0 \leq a_j \leq q-1$.
%\end{itemize}

\begin{lemma}\label{LEM:15}
For $t, s \geq 1$, $\lambda \mid (q-1)$, and $0 < \ell_0 < \frac{q-1}{\lambda}$,
$$\left|\left\{(i_1,i_2,\ldots, i_t):~\sum_{i=1}^t i_i \equiv 0~({\rm mod}~\lambda), \begin{array}{c} 
q-\lambda \ell_0\leq  i_j \leq q-1,~\forall ~1 \leq j \leq s \\ 
1 \leq  i_j \leq q-1,~\forall ~s+1 \leq j \leq t 
\end{array}  \right \}\right|=\frac{(\lambda \ell_0)^s(q-1)^{t-s} }{\lambda}.$$	
\end{lemma}

\begin{proof}
The proof follows a similar argument to Lemma \ref{LEM:11} and is omitted here for concision. 	
\end{proof}

\begin{lemma}\cite{Srensen}\label{LEM:16}
For integers $t, s, l\geq 0$,
$$\left|\{(x_1,x_2,\ldots,x_t):~0\leq x_i\leq l,~x_1+x_2+\cdots+x_t=s\}\right|=\sum_{j=0}^t (-1)^j\binom{t}{j}\binom{s-j(l+1)+t-1}{s-j(l+1)}.$$
\end{lemma}

\begin{theorem}\label{THM::17}
Let $1\leq \ell_1\leq m-1$ and $0<\ell_0<\frac{q-1}{\lambda}$ with $\lambda \mid(q-1)$. Then
$$|\cB|=\sum_{t=1}^m\sum_{s=1}^t \Xi_{s,t} \frac{(\lambda \ell_0)^s (q-1)^{t-s}}{\lambda}$$
where 
\begin{align*}
\Xi_{s,t} &=(\ell_1+1) \binom{t-1}{s-1} \sum_{j=0}^{t-s} (-1)^j\binom{t-s}{j}\binom{m-s(\ell_1+1)-1-j \ell_1}{m-t-(s+j)\ell_1} \\
&~~~+ \binom{t-1}{s} \sum_{j=0}^{t-s-1}(-1)^j \binom{t-s-1}{j} \sum_{s'=0}^{\ell_1-1} (s'+1) \binom{m-s(\ell_1+1)-2-s'-j\ell_1 }{m-t-(s+j)\ell_1-s'}.
\end{align*}	
\end{theorem}

\begin{proof}
 Recall that the weight of $\overline{i}=(i_{m-1}, i_{m-2},\ldots,i_0)$ is defined by
 $$\wt(\overline{i})=\left|\left\{0\leq j\leq m-1:~i_j\neq 0\right\}\right|.$$
 Partition the set $\cB$ by weight: $\cB = \cup_{t=1}^m \cB_t$, where $\cB_t = \{ i \in \cB : \wt(\overline{i}) = t \}$. For $i \in \cB_t$, represent $\overline{i}$ in the canonical form
$$\overline{i}=(\underbrace{0,\ldots, 0}_{x_1},a_1,\underbrace{0,\ldots, 0}_{x_2},a_2, \ldots, \underbrace{0,\ldots, 0}_{x_t},a_t, \underbrace{0,\ldots, 0}_{x_{t+1}}),$$
where each $a_j$ is an integer satisfying $1\leq a_j \leq q-1$, each $x_j$ is a non-negative integer, and $x_1+x_2+\cdots+x_{t+1}=m-t$. %Then 
%$$\mathtt{Run}(\overline{i})=\max\left\{ x_1+x_{t+1}, x_2, \cdots,x_{t} \right\}.$$
Note that
$$\overline{\lambda\delta} = (\underbrace{0,\dots,0}_{\ell_1}, q-1-\lambda\ell_0,\underbrace{q-1,\dots,q-1}_{m-1-\ell_1}).
$$
Define $$v_j=\begin{cases}
x_1+x_{t+1}~&{\rm if}~j=1,\\
x_j~&{\rm if}~2\leq j\leq t.	
\end{cases}
$$
Then $i\in \cB_t$ if and only if the following conditions hold:
\begin{enumerate}
	\item $\max\{ v_1, v_2, \cdots,v_{t} \}=\ell_1$.
	\item $v_1+\cdots+v_t=m-t$.
	\item $a_1+a_2+\cdots+a_t \equiv 0~({\rm mod}~\lambda)$.
	\item $q-\lambda \ell_0 \leq a_j \leq q-1$ if $v_j=\ell_1$.
\end{enumerate}
Define $K=\{ 1\leq j\leq t:~v_j \}=\ell_1$ ad let $s=|K|$.  We distinguish cases by the membership of index $1$ in $K$.

{\it Case 1}: $1\in K$, i.e., $v_1=x_1+x_{t+1}=\ell_1$. The constraints $0\leq x_1\leq \ell_1$ and $x_{t+1}=\ell_1-x_1$ yield $\ell_1+1$ solutions for $(x_1,x_{t+1})$. Since $|K|=s$ and $1\in K$, select $s-1$ indices from $\{2,3, \cdots, t\}$ for $K$ in $\binom{t-1}{s-1}$ ways. For $j\notin K$, the sum constraint is
$$\sum_{1\leq i\leq t, i\notin K} v_i=m-t-s\ell_1$$
with $0\leq v_j\leq \ell_1-1$.  The coefficients $a_j$ satisfy 
$$a_1+a_2+\cdots+a_t \equiv 0~({\rm mod}~\lambda),$$ 
where for $j\in K$, $q-\lambda \ell_0\leq  a_j\leq q-1$, and for $j\notin K$, $1\leq  a_j\leq q-1$. By Lemma \ref{LEM:15} and Lemma \ref{LEM:16}, the number of such vectors is 
$$(\ell_1+1) \binom{t-1}{s-1} \sum_{j=0}^{t-s} (-1)^j\binom{t-s}{j}\binom{m-s(\ell_1+1)-1-j \ell_1}{m-t-(s+j)\ell_1} \frac{(\lambda \ell_0)^s  (q-1)^{t-s}}{\lambda } .$$

{\it Case 2}: $1\notin K$, i.e., $0\leq v_1= x_1+x_{t+1}\leq \ell_1-1$. Select all $s$ elements of $K$ from $\{2,\ldots,t\}$ in $\binom{t-1}{s}$ ways. For fixed $v_1$, the equations $0 \leq x_1 \leq v_1$ and $x_{t+1} = v_1 - x_1$ yield $v_1 + 1$ solutions. The sum constraint for $j \geq 2$ not in $K$ is
$$\sum_{2\leq i\leq t, i\notin K} v_i=m-t-s\ell_1-v_1$$
with $0\leq v_j\leq \ell_1-1$. The conditions on $a_j$ are identical to Case 1. By Lemma \ref{LEM:15} and Lemma \ref{LEM:16}, the number of such vectors is
$$ \binom{t-1}{s} \Delta \frac{(\lambda \ell_0)^s  (q-1)^{t-s}}{\lambda },$$
where 
\begin{align*}
\Delta= \sum_{v_1=0}^{\ell_1-1} (v_1+1)\sum_{j=0}^{t-s-1}(-1)^j \binom{t-1-s}{j} \binom{m-s(\ell_1+1)-2-v_1-j\ell_1 }{m-t-(s+j)\ell_1-v_1}.
\end{align*}
Summing both cases over $s \in \{1,2,\cdots,t\}$ gives $|\mathcal{B}_t|$, and $|\mathcal{B}| = \sum_{t=1}^m |\mathcal{B}_t|$.  This completes the proof.
\end{proof}

Combining Theorem \ref{THM:13}, Theorem \ref{THM::17} and \eqref{eq:dim3}, we obtain

\begin{theorem}\label{THM:18}
	Let $1\leq \ell_1\leq m-1$ and $0<\ell_0<\frac{q-1}{\lambda}$ with $\lambda \mid(q-1)$. Then
	\begin{align*}
	&\dim(\C_{(q, m, \lambda, \ell_0, \ell_1)})\\
	=~&\frac{q^m-1}{\lambda}-\left(\frac{q-1}{\lambda}\right)\sum_{i=1}^{\lfloor  \frac{m}{\ell_1+1}    \rfloor} (-1)^{i-1} \frac{m(q-1)^{i-1}}{i}  \binom{m-i \ell_1-1}{i-1}  q^{m-i(\ell_1+1)}+m+|\cB|,	
	\end{align*}
	with $|\cB|$ as in Theorem \ref{THM::17}.
\end{theorem}

Using the computational algebra system MAGMA, we verified the assertions of Theorem \ref{THM::17}. The parameters of the BCH code $\mathcal{C}_{(q, m, \lambda, \ell_0, \ell_1)}$ are tabulated in Table \ref{table3}.

\begin{table*}
{\caption{\rm Parameters of $\C_{(q, m, \lambda, \ell_0, \ell_1)}$ for $1\leq \ell_1\leq m-1$ and $1\leq  \ell_0<\frac{q-1}{\lambda}$
}\label{table3}
\begin{center}
\begin{tabular}{ccccccc}\hline
$q$  & $m$ & $\lambda$  & $\ell_0$&  $\ell_1$ & Parameters & Optimality \\ \hline
$3$ & $3$ &  $1$ & $1$ & $1$ & $[26,17,5]$ & $d_{\rm optimal}=6$\\ \hline
$3$ & $4$ &  $1$ & $1$ & $1$ & $[80,38,17]$ &$d_{\rm best}=21$\\ \hline
$3$ & $4$ &  $1$ & $2$ & $1$ & $[80,60,8]$ & Best Known\\ \hline
$3$ & $4$ &  $1$ & $1$ & $2$ & $[80,68,5]$ & $d_{\rm optimal}=6$\\ \hline
$3$ & $4$ &  $1$ & $2$ & $2$ & $[80,76,2]$ & Optimal\\ \hline
$3$ & $5$ &  $1$ & $1$ & $1$ & $[242,87,53]$ & $d_{\rm best}=59$\\ \hline
$3$ & $5$ &  $1$ & $2$ & $1$ & $[242,157,26]$ & Best Known\\ \hline
$3$ & $5$ &  $1$ & $1$ & $2$ & $[242,187,17]$ & Best Known\\ \hline
$3$ & $5$ &  $1$ & $2$ & $2$ & $[242,217,8]$ & Best Known\\ \hline
$3$ & $5$ &  $1$ & $1$ & $3$ & $[242,227,5]$ & Best Known\\ \hline
$3$ & $5$ &  $1$ & $2$ & $3$ & $[242,237,2]$ & Optimal\\ \hline
$5$ & $2$ &  $1$ & $1$ & $1$ & $[24,20,3]$ & $d_{\rm optimal}=4$\\ \hline
$5$ & $2$ &  $1$ & $2$ & $1$ & $[24,22,2]$ & Optimal\\ \hline
$5$ & $3$ &  $1$ & $1$ & $1$ & $[124,79,19]$ & Best Known\\ \hline
$5$ & $3$ &  $1$ & $2$ & $1$ & $[124,91,14]$ & Best Known\\ \hline
$5$ & $3$ &  $1$ & $3$ & $1$ & $[124,103,9]$ & $d_{\rm best}=10$\\ \hline
$5$ & $3$ &  $1$ & $1$ & $2$ & $[124,118,3]$ & $d_{\rm optimal}=4$\\ \hline
$5$ & $3$ &  $1$ & $2$ & $2$ & $[124,121,2]$ & Optimal\\ \hline
$5$ & $3$ &  $2$ & $1$ & $1$ & $[62,47,7]$ & $d_{\rm best}=8$\\ \hline
\end{tabular}
\end{center}}
\end{table*}

Although the general dimension formulas in Theorems \ref{THM::14} and \ref{THM:18} exhibit complexity, they simplify significantly under specific parameter constraints. When $1 \leq \ell_1 \leq \lfloor (m-2)/2 \rfloor$ and $\lambda = q-1$, Theorem \ref{THM::14} resolves Open Problem \ref{Open3} posed by Li et al. \cite{BCH5}. For the complementary range $\lceil (m-1)/2 \rceil \leq \ell_1 \leq m-1$, we derive a simplified closed-form expression. The derivation relies on the following auxiliary results.

\begin{lemma}\cite{BCH6}\label{lem:6}
Let $m\geq 2$ and let $1\leq i\leq q^{\lfloor (m+1)/2\rfloor}-1$ with $i\not\equiv 0\pmod q $. Then  $i$ is a $q$-cyclotomic coset leader modulo $N$ (i.e., $i \in \Gamma_{(q, N)}$) and $|C_i^{(q, N)}| = m$.
\end{lemma}

\begin{lemma}\cite[Lemma 1]{SLZT}\label{lem:7}
    Let $n=(q^m-1)/\lambda$ and let $\Gamma_{(q, n)}$ be the set of $q$-cyclotomic coset leaders modulo $n$. Then for any integer $i$, $i \in \Gamma_{(q, n)} \iff \lambda i \in \Gamma_{(q, N)}$  and $|C_i^{(q, n)}| = |C_{\lambda i}^{(q, N)}|$.
\end{lemma}

\begin{theorem}\label{THM:2}
Let $m\geq 2$ and $\lceil (m-1)/2\rceil\leq   \ell_1\leq m-1$. Then 
$\dim( \C_{(q, m, \lambda, \ell_0, \ell_1)} )=\frac{q^m-1}{\lambda}-m\cdot \varepsilon$, where
$$\varepsilon=\begin{cases}
\frac{q-1}{\lambda}-1-\ell_0 ~&{\rm if}~\ell_1=m-1,\\
(\frac{q-1}{\lambda}) (q - \lambda \ell_0) q^{m-2-\ell_1} - 1~&{\rm if}~\lceil (m-1)/2\rceil\leq   \ell_1\leq m-2.
\end{cases}
 $$
\end{theorem}
\begin{proof}
From \eqref{Delta}, $\lambda \delta \leq q^{m-\ell_1} - 1$. For $1 \leq i \leq \delta - 1$ with $i \not\equiv 0 \pmod{q}$, we have $$\lambda i \leq \lambda(\delta - 1) < q^{\lfloor (m+1)/2 \rfloor} - 1.$$ By Lemma~\ref{lem:6}, $\lambda i$ is a coset leader with $|C_{\lambda i}^{(q, N)}| = m$. Lemma~\ref{lem:7} then implies $i$ is a coset leader modulo $n$ with $|C_i^{(q, n)}| = m$. The number of such $i$ is $\delta - 1 - \lfloor (\delta - 1)/q \rfloor$, so
	\[
	\dim(\C) = n - m \left( \delta - 1 - \left\lfloor \frac{\delta - 1}{q} \right\rfloor \right).
	\]
	Substituting $\delta$ from \eqref{Delta} yields $\varepsilon$. This completes the proof.
\end{proof}

%Verification through MAGMA computations confirms the theoretical results, while Tables \ref{table1} and \ref{table2} demonstrate that the BCH code $\mathcal{C}_{(q, m, \lambda, \ell_0, \ell_1)}$ achieves competitive parameters.

%MAGMA computations verify our results, and the BCH code $\C_{(q, m, \lambda, \ell_0, \ell_1)}$ has good parameters as documented in Tables \ref{table1} and \ref{table2}.

\section{Conclusion}

The main contribution of this paper is the determination of the parameters of the BCH codes $\C_{(q, m, \lambda, \ell_0, \ell_1)}$. We proved that their minimum distance equals their designed distance (see Theorem \ref{Distance}). This resolved two open problems posed in  \cite{BCH5}. Furthermore, we derived closed-form dimension formulas via combinatorial enumeration techniques (see Theorem \ref{THM::14} and Theorem \ref{THM:18}), resolving a third open problem from \cite{BCH5}. Although the codes $\C_{(q, m, \lambda, \ell_0, \ell_1)}$ have flexible parameters, characterizing the minimum distance for broader families of BCH codes remains an open challenge.


\begin{thebibliography}{99}

\bibitem{AK98}  
E. F. Assmus Jr. and J. D. Key, ``Polynomial codes and finite geometries," in \emph{Handbook of Coding Theory}, eds. V. S. Pless and W. C. Huffman. Amsterdam: Elsevier, 1998, pp. 1269--1343.

\bibitem{Diem1}
E.R. Berlekamp, ``The enumeration of information symbols in BCH codes," \emph{Bell Syst. Tech. J.}, pp. 1861--1880, 1967.

\bibitem{Min1}
E.R. Berlekamp, ``The weight enumerator of certain subcodes of the second order binary Reed-Muller codes," \emph{Inf. Control}, vil. 17, pp. 485--500, 1970.

\bibitem{BST}
A. Berman, Y. Shany, and I. Tamo,``Efficient algorithms for constructing minimum-weight codewords in some extended binary BCH codes,"  \textit{IEEE Trans. Inf. Theory}, vol. 70, no. 11, pp. 7673--7689, Nov. 2024. 


\bibitem{BCH2}
R. C. Bose, D. K. Ray-Chaudhuri, ``On a class of error correcting binary group codes," \textit{Inf.  Control}, vol. 3, pp. 68--79, 1960.

\bibitem{Charpin}
P. Charpin, ``Open problems on cyclic codes," in Handbook Coding Theory, vol. 1, V. Pless and W. C. Huffman, Eds. Amsterdam, The Netherlands: Elsevier, 1998, ch. 11, pp. 963-1063.

\bibitem{CJLM}
A. Cherchem, A. Jamous, H. Liu and Y. Maouche, ``Some new results on dimension and Bose distance for various classes of BCH codes," \emph{Finite Fields Appl.}, vol. 65, 101673, 2020.

\bibitem{DGM70}
P. Delsarte, J. M. Goethals, and F. J. MacWilliams,``On generalized Reed-Muller codes and their relatives," \emph{Inf. Control}, vol. 16, no. 5, pp. 403--442, July 1970.

\bibitem{Ding-1}
C. Ding, Codes From Difference Sets. Singapore: World Scientific, 2015.

\bibitem{Ding}
C. Ding, ``Parameters of several classes of BCH codes," \emph{IEEE Trans. Inf. Theory}, vol. 61, no. 10, pp. 5322--5330, Oct 2015.

\bibitem{DDZ} 
C. Ding, X. Du, and Z. Zhou, ``The Bose and minimum distance of a class of BCH codes," \emph{IEEE Trans. Inf. Theory}, vol. 61, no. 5, pp. 2351--2356, May 2015.

\bibitem{DL}
C. Ding, and C. Li, ``BCH cyclic codes," \emph{Discrete Math.}, vol. 347, 113918, 2024.

\bibitem{DLX18} 
C. Ding, C. Li, and Y. Xia, ``Another generalisation of the binary Reed-Muller codes and its applications," \emph{Finite Fields Appl.}, vol. 53, pp. 144--174, Sep. 2018.

\bibitem{DFZ}
C. Ding, C. Fan, and Z. Zhou, ``The dimension and minimum distance of two classes of primitive BCH codes," \emph{Finite Fields Appl.}, vol. 45, pp. 237--263, 2017.

\bibitem{DLMQ}
S. Dong, C. Li, S. Mesnager, and H. Qian,``Parameters of squares of primitive narrow-sense BCH codes and their complements," \textit{IEEE Trans. Inf. Theory}, vol. 69, no. 8, pp. 5017--5031, Aug. 2023. 

\bibitem{GLQS}
C. Gan, C. Li, H. Qian, and X. Shi, ``On Bose distance of a class of BCH codes with two types of designed distances," \textit{Des. Codes Cryptogr.}, vol. 92, pp. 2031--2053, 2024.

\bibitem{Grassl}
M. Grassl, Bounds on the minimum distance of linear codes and Quantum Codes. http://www.codetables.de. 

\bibitem{Griesmer}
J. H. Griesmer, ``A bound for error-correcting codes," \textit{IBM Journal of Research and Development}, vol. 4, pp. 532--542, 1960.

\bibitem{HYWSM}
X. Huang, Q. Yue, Y. Wu, X. Shi, and J. Michel, ``Binary primitive LCD BCH codes," \textit{Des. Codes Cryptogr.}, vol. 88, pp. 2453--2473, 2020.

\bibitem{BCH1}
A. Hocuenghem, ``Codes correcteurs d'erreurs," \textit{Chiffers}, vol. 2, pp. 147--156, 1959.

\bibitem{KL}
T. Kasami and S. Lin, ``Some results on the minimum weight of primitive BCH codes," \textit{IEEE Trans. Inf. Theory}, vol. 18, pp. 824--825, 1972.

\bibitem{LI}
S. Li, ``The Minimum distance of some narrow-sense primitive BCH codes," \textit{SIAM J. Discrete Math.}, vol. 31, no. 4, pp. 2530--2569, 2017.

\bibitem{BCH4}
C. Li, C. Ding, ``LCD cyclic codes over finite fields," \textit{IEEE Trans. Inf. Theory}, vol. 63, no. 7, pp. 4344--4356, Jul. 2017.

\bibitem{BCH5}
S. Li, C. Ding, M. Xiong, and G. Ge, ``Narrow-sense BCH codes over $\gf(q)$ with length $n=(q^m-1)/(q-1)$," \textit{IEEE Trans. Inf. Theory}, vol. 63, no. 11, pp. 7219--7236, Nov. 2017.

\bibitem{LLDL}
S. Li, C. Li, C. Ding, and H. Liu, ``Two families of LCD BCH codes," \textit{IEEE Trans. Inf. Theory}, vol. 63, no. 9, pp. 5699--5717, Sep. 2017. 

\bibitem{LWL}
C. Li, P. Wu, and F. Liu,``On two classes of primitive BCH codes and some related codes," \textit{IEEE Trans. Inf. Theory}, vol. 65, no. 6, pp. 3830--3840, Jun. 2019.


\bibitem{LYW}
F. Li, Q. Yue, and Y. Wu, ``Designed distances and parameters of new LCD BCH codes over finite fields," \textit{Cryptogr. Commun.}, vol. 12, pp. 147--163, 2020.

\bibitem{BCH7}
X. Ling, S. Mesnager, Y. Qi, C. Tang, ``A class of narrow-sense BCH codes over $\gf(q)$ of length $n=(q^m-1)/2$," \textit{Des. Codes Cryptogr.}, vol. 88, pp. 413--427, Feb. 2020.

\bibitem{BCH6}
H. Liu, C. Ding, C. Li, ``Dimensions of three types of BCH codes over $\gf(q)$," \textit{Discrete Math.}, vol. 340, no. 8, pp. 1910--1927, Aug. 2017.

\bibitem{LLFLR}
Y. Liu, R. Li, Q. Fu, L. Lu, and Y. Rao, ``Some binary BCH codes with length $n=2^m+1$," \textit{Finite Fields Appl.}, vol. 55, pp. 109--133, 2019.

\bibitem{LLGS}
Y. Liu, R. Li, L. Guo, and H. Song, ``Dimensions of nonbinary antiprimitive BCH codes and some
conjectures," \textit{Discrete Math.}, vol. 346, 113496, 2023.


\bibitem{mann}
H. B. Mann, ``On the number of information symbols in Bose-Chaudhuri codes," \textit{Inf. Control}, vol. 5, no. 2, pp. 153--162, 1962.


\bibitem{NLJM}
S. Noguchi, X. Lu, M. Jimbo, and Y. Miao, ``BCH codes with minimum distance proportional to code length," \textit{SIAM J. Discrete Math.}, vol. 35, no. 1, pp. 179--193, 2021.

\bibitem{PZK}
B. Pang, S. Zhu, and X. Kai, ``Five families of the narrow-sense primitive BCH codes over finite fields," \textit{Des. Codes Cryptogr.}, vol. 89, pp. 2679--2696, 2021.

\bibitem{SB}
Y. Shany and A. Berman,``The generating idempotent is a minimum-weight codeword for some binary BCH codes," \textit{IEEE Trans. Inf. Theory}, vol. 71, no. 3, pp. 1700--1704, Mar. 2025.

\bibitem{Srensen}

A. Srensen, ``Projective Reed-Muller codes," \textit{IEEE Trans. Inf. Theory}, vol. 37, no. 6, pp. 1567--1576, 1991.

 
\bibitem{SLZT}
Z. Sun, X. Liu, S. Zhu, and Y. Tang, ``Negacyclic BCH codes of length $(q^{2m}-1)/(q+1)$ and their duals," \textit{Des. Codes Cryptogr.}, vol. 92, pp. 2085--2101, 2024.

\bibitem{WWLW}
X. Wang, J. Wang, C. Li, and Y. Wu, ``Two classes of narrow-sense BCH codes and their duals," \textit{IEEE Trans. Inf. Theory}, vol. 70, no. 1, pp. 131--144, Jan. 2024.

\bibitem{XWLC}
H. Xu, X. Wu, W. Lu, and X. Cao, ``The sufficient and necessary conditions for the minimum distance of the BCH code $\C_{(q, q+1, 3, h)}$ to be $3$ and $4$", \textit{IEEE Trans. Inf. Theory}, vol. 71, no. 6, pp. 4206--4213, Jun. 2025.

\bibitem{ZSWH}
H. Zhu, M. Shi, X. Wang, and Tor Helleseth, ``The $q$-ary antiprimitive BCH codes," \textit{IEEE Trans. Inf. Theory}, vol. 68, no. 3, pp. 1683--1695, Aug. 2022. 

\bibitem{ZSK}
S. Zhu, Z. Sun, and X. Kai, ``A class of narrow-sense BCH codes," \textit{IEEE Trans. Inf. Theory}, vol. 65, no. 8, pp. 4699--4714, Aug. 2019. 


%\bibitem{C99}
%R. S. Coulter, On the evaluation of a class of Weil sums in characteristic $2$, \emph{New Zealand J. of Math.} 28(1999), 171-184.




\end{thebibliography}
\end{document}